\documentclass[12pt]{amsart}
\usepackage{amsmath}
\usepackage{amsfonts}
\usepackage{amssymb}
\usepackage{amsthm}
\usepackage[all]{xy}
\usepackage{color}

\newtheorem{theorem}{Theorem}[section]
\newtheorem{lemma}[theorem]{Lemma}
\newtheorem{proposition}[theorem]{Proposition}
\newtheorem{corollary}[theorem]{Corollary}
\newtheorem{definition}[theorem]{Definition}
\newtheorem{remark}[theorem]{Remark}
\newtheorem{example}[theorem]{Example}
\newcommand{\RN}[1]{%
  \textup{\uppercase\expandafter{\romannumeral#1}}%
}

\numberwithin{equation}{section}

 \author[Nupur Patanker]{Nupur Patanker}
 \address{ Indian Institute of Science Education and Research, Bhopal} 
 \email{nupurp@iiserb.ac.in}
\author[Sanjay Kumar Singh]{Sanjay Kumar Singh}
 \address{Indian Institute of Science Education and Research, Bhopal}
  \email{sanjayks@iiserb.ac.in}
\title[Geometric Goppa codes over Elementary Abelian $p$-Extensions of $\mathbb{F}_{p^{s}}(x)$]{On geometric Goppa codes from Elementary Abelian $p$-Extensions of $\mathbb{F}_{p^{s}}(x)$}

\date{}
\begin{document}
\begin{abstract}
Let $p$ be a prime number and $s> 0$ an integer. In this short note, we investigate one-point geometric Goppa codes associated with an elementary abelian $p$-extension of $\mathbb{F}_{p^{s}}(x)$. We determine their dimension and the exact minimum distance in a few cases. These codes are a special case of weak Castle codes. We also list the exact values of the second generalized Hamming weight of these codes in a few cases. Simple criteria for the self-duality and the quasi-self-duality of these codes are also provided. Furthermore, we construct examples of quantum codes, convolutional codes, and locally recoverable codes on the function field.
\end{abstract}

\maketitle
\section{\textbf{Introduction}}
Let $\mathbb{F}_{p^{s}}$ be the finite field with $p^{s}$ elements of characteristic $p$ (where $s$ is a positive integer). A linear code is a $\mathbb{F}_{p^{s}}$-subspace of $\mathbb{F}_{p^{s}}^{n}$, the $n$-dimensional standard vector space over $\mathbb{F}_{p^{s}}$. Such codes are utilized for the transmission of information. \par It was observed by Goppa in \cite{Goppa} that we can use divisors in a field of algebraic functions to construct a class of linear codes. In Goppa's construction, we choose a divisor $G$ and $n$ rational places ( i.e.\ places of degree one) of the algebraic function field to form a linear code of length $n$. These codes are called geometric Goppa codes. If $G$ is of the form $rQ$, for a place $Q$ of the algebraic function field and integer $r$, then these codes are called one-point codes. One-point geometric Goppa codes over algebraic function field of Hermitian curve have been studied in \cite{book}, \cite{hermi}, \cite{hermi2}, \cite{tru}, etc. \par
Elementary abelian $p$-extension of the rational function field $\mathbb{F}_{{p}^{s}}(x)$ is a Galois extension $F$ of $\mathbb{F}_{p^s}(x)$ such that Gal$(F/\mathbb{F}_{p^s}(x))$ is an elementary abelian group of exponent $p$. The algebraic function field associated with Hermitian curve is an example of such an extension. The properties of elementary abelian $p$-extensions of $\mathbb{F}_{{p}^{s}}(x)$ have been studied in \cite{book}, \cite{elem4}, \cite{elem3}, \cite{elem}, etc. Another example of elementary abelian $p$-extension of $\mathbb{F}_{p^s}(x)$ is the function field $\mathbb{F}_{p^s }(x,y)/\mathbb{F}_{p^s}$ defined by $A(y)=B(x)$, where $A(T) \in \mathbb{F}_{{p}^{s}}[T]$ is a separable, additive polynomial of degree $q = p^t$, for some $t$, such that all its roots are contained in $\mathbb{F}_{p^s}$, and the degree of $B(T) \in \mathbb{F}_{{p}^{s}}[T]$ is not divisible by $p$. The non-singular projective curve $\mathcal{X}$ associated with this function field was studied in \cite{elem}. In \cite{elem}, T. Johnsen, S. Manshadi and N. Monzavi determined parameters of geometric Goppa codes $C_{\mathcal{L}}(D,G)$ on this function field with the assumption that $deg~D \geq 4g-2$. In this note, we study one-point geometric Goppa codes on an elementary abelian $p$-extension of $\mathbb{F}_{p^{s}}(x)$ without the above assumption. 

The special type of elementary abelian $p$-extensions of $\mathbb{F}_{p^{s}}(x)$ considered in the present note are examples of function fields associated with weak Castle curves. Castle curves and weak Castle curves are of interest for coding theory purposes. Many known codes belong to the class of Castle and weak Castle codes. Castle curves, weak Castle curves, and codes associated with them have been studied in \cite{weakcastle1}, \cite{weakcastle2}, \cite{weakcastle3}, \cite{weakcastle4}, etc. In \cite{weakcastle1}, C. Munuera, A. Sep\'{u}lveda and F. Torres have studied one-point geometric Goppa codes arising from Castle and weak Castle curves. The authors have determined bounds on the minimum distance and generalized Hamming weights of these codes. In \cite{weakcastle2}, Wilson Olaya-Le\'{o}n and C. Munuera determined order-like bound $d^{*}$ on the minimum distance of some Castle codes ( i.e.\ one-point geometric Goppa codes from Castle curves), particularly, those related to semigroups generated by two elements and telescopic semigroups. In \cite{weakcastle3}, Wilson Olaya-Le\'{o}n and Claudia Granados-Pinz\'{o}n computed the bound $d_{2}^{*}$ on the second generalized Hamming weight for some Castle codes. In \cite{weakcastle4}, C. Munuera, W. Ten\'{o}rio and F. Torres studied geometric Goppa codes producing quantum codes. The authors paid particular attention to the family of Castle and weak Castle codes. In \cite{quantum3}, F. Hernando, G. McGuire, F. Monserrat and J. J. Moyano-Fernndez have obtained new quantum codes with good parameters which are constructed from self-orthogonal geometric Goppa codes over function fields associated to a wide class of curves. Our goal in this note is to determine the exact value of the second generalized Hamming weight of one-point geometric Goppa codes defined over elementary abelian extensions of $\mathbb{F}_{p^{s}}(x)$.\par 

This note is organized as follows. In section $2$, we recall some results about Goppa's construction of linear codes and generalized Hamming weights of linear codes. In section $3$, we study the properties of elementary abelian $p$-extension $F/ \mathbb{F}_{p^s}$. In section $4$, we define  one-point geometric Goppa codes over this function field and study its parameters. In section $5$, we list the second generalized Hamming weight of these codes.  In section $6$, we determine simple conditions for the self-duality and the quasi-self-duality of these codes. In section $7$ and $8$, we obtain examples of quantum codes and convolutional codes from one-point geometric Goppa codes constructed in section $4$. In section $9$, we obtain locally recoverable codes on the function field $F/ \mathbb{F}_{p^s}$.

\section{\textbf{Preliminaries}}
\subsection{Geometric Goppa code (\cite{book}, Chapter $2$)}
Goppa's construction of linear codes over $\mathbb{F}_{p^{s}}$ is described as follows:\par
Let $F' / \mathbb{F}_{p^{s}}$ be an algebraic function field of genus $g'$. Let $P_{1},\cdots,P_{n}$ be pairwise distinct places of $F'/ \mathbb{F}_{p^{s}}$ of degree one. Let $D':=P_{1}+ \cdots +P_{n}$ and $G$ be a divisor of $F' / \mathbb{F}_{p^{s}}$ such that $supp(G) \cap supp(D')=\emptyset$. The geometric Goppa code $C_{\mathcal{L}}(D',G)$ associated with $D'$ and $G$ is defined as
$$C_{\mathcal{L}}(D',G):=\{(x(P_{1}), \cdots ,x(P_{n})):~ x \in \mathcal{L}(G)\} \subseteq \mathbb{F}_{p^{s}}^{n}.$$
Thus, $C_{\mathcal{L}}(D',G)$ is an $[n,k,d]$ code with parameters $k=dim(\mathcal{L}(G))-dim(\mathcal{L}(G-D'))$ and $d \geq n-deg(G)$.\par
Another code associated with the divisors $G$ and $D'$ is defined using local components of Weil differentials. The code $C_{\Omega}(D',G) \subseteq \mathbb{F}_{p^{s}}^{n}$ is defined as
$$C_{\Omega}(D',G):=\{(\omega_{P_{1}}(1),\cdots,\omega_{P_{n}}(1)):~\omega  \in \Omega_{F'}(G-D')\}.$$
Thus, $C_{\Omega}(D',G)$ is an $[n,k,d]$ code with parameters $k=i(G-D')-i(G)$ and $d \geq deg(G)-(2g'-2).$\par
$C_{\Omega}(D',G)$ is the dual code of $C_{\mathcal{L}}(D',G)$ with respect to Euclidean scalar product on  $\mathbb{F}_{p^{s}}^{n}$ i.e.\ $C_{\Omega}(D',G)=C_{\mathcal{L}}(D',G)^{\perp}.$ Let $\eta$ be a Weil differential of $F'$ such that $\nu_{P_{i}}(\eta)=-1$ and $\eta_{P_{i}}(1)=1$ for $i=1,\cdots,n$. Then, $C_{\mathcal{L}}(D',G)^{\perp}=C_{\Omega}(D',G)=C_{\mathcal{L}}(D',D'-G+(\eta))$.

\subsection{Generalized Hamming weights of Linear codes}
The support of a $[n,k]$ linear code $C$ over $\mathbb{F}_{{p}^{s}}$ is defined by 
$$supp(C):=\{i~:~ x_{i} \neq 0 \text{ for some } \mathbf{x}=(x_{1},\cdots,x_{n}) \in C\}.$$
For $1 \leq l \leq k$, the \textit{$l$-th generalized Hamming weight} of $C$ is defined by 
$$d_{l}(C):=min\{~ |supp(D)| ~:~ D \text{ is a linear subcode of } C \text { with } dim(D)=l\}.$$
In particular, the first generalized Hamming weight of $C$ is the usual minimum distance. The \textit{weight hierarchy} of code $C$ is the set $\{d_{1}(C), \cdots, d_{k}(C)\}$ of generalized Hamming weights. The generalized Hamming weights for linear codes were introduced in \cite{ghw1}, \cite{ghw2}, and rediscovered by Wei in \cite{ghw3}. The study of these weights was motivated by some applications in cryptography.\par
Few properties of generalized Hamming weights of $C$ have been listed in the following theorems.
\begin{theorem}{$($\cite{ghw3}, Theorem $1)$}
For an $[n,k]$ linear code $C$ with $k>0$, we have 
$$1 \leq d_{1}(C) < d_{2}(C)< \cdots <d_{k}(C) \leq n.$$
\end{theorem}

Let $\mathbf{H}$ be a parity check matrix of $C$, and let $\mathbf{h}_{i}$, $1 \leq i \leq n$, be its column vectors. For $I \subseteq \{1, \cdots, n \}$, let $\langle \mathbf{h}_{i}: i \in I \rangle$ denotes the space generated by those vectors.
Then
\begin{theorem}{$($\cite{ghw3}, Theorem $2)$}
$d_{l}(C)=min\{~ |I|:~ |I|-rank(\langle \mathbf{h}_{i}: i \in I \rangle)\geq l\}.$
\end{theorem}

For geometric Goppa code $C_{\mathcal{L}}(D',G)$, the $l$-th generalized Hamming weight is given by the following theorem.
\begin{theorem}{$($\cite{ghwag}, Corollary $1)$}
Let $C=C_{\mathcal{L}}(D',G)$ be a code of dimension $k$ and $a:=dim(\mathcal{L}(G-D')) \geq 0$. Then for every $l$, $ 1 \leq l \leq k$,
\begin{align*}
d_{l}(C)&=min\{deg(D'')~:~0 \leq D'' \leq D',~ dim(\mathcal{L}(G-D'+D'')) \geq l+a\}\\
&=min\{n-deg(D'')~:~0 \leq D'' \leq D',~ dim(\mathcal{L}(G-D'')) \geq l+a\}.
\end{align*}
\end{theorem}

\subsubsection{Feng-Rao distances on numerical semigroups}~\\
The Feng-Rao distances on numerical semigroups are defined in \cite{semi}. We will explain it briefly in this subsection.\par
Let $A$ be a numerical semigroup. If for a set $G \subseteq A$, every $x \in A$ can be written as a linear combination $$x=\sum_{y \in G} \lambda_{y} y,$$
where finitely many $\lambda_{y} \in \mathbb{N} \cup \{0\}$ is non-zero, then we say that $A$ is generated by $G$. It is well known that every numerical semigroup is finitely generated. An element $x \in A$ is said to be irreducible if $x=a+b$ for $a,b \in A$ implies $a.b=0$. Every generator set contains the set of irreducible elements and the set of irreducibles generates $A$. The number of irreducible elements is called the \textit{embedding dimension} of $A$. We enumerate the elements of $A$ in increasing order $$A=\{\rho_{1}=0 < \rho_{2}< \cdots\}.$$
For $a,b \in \mathbb{Z}$ given, we say that $a$ \textit{divides} $b$, and write
$$a \leq_{A} b, \text{ if } b-a \in A.$$
The binary relation is an order relation.\par
The set $D(y)$ denotes the set of \textit{divisors} of $y$ in $A$, and for given $M=\{m_{1}, \cdots, m_{t} \} \subseteq A$, we write $D(M)=D(m_{1}, \cdots, m_{t})=\bigcup^{t}_{i=1} D(m_{i})$.

\begin{definition}
Let $A$ be a numerical semigroup, that is, a submonoid of $\mathbb{N}$ such that $|\mathbb{N} \setminus A | < \infty$ and $0 \in A$. We call $g:=|\mathbb{N}\setminus A|$ the \textit{genus} of $A$. The unique element $c \in A$ such that $c-1 \not \in A$ and $c+u \in A$ for all $u \in \mathbb{N}$  is called the \textit{conductor} of $A$. The (classical) Feng-Rao distance of $A$ is defined by the function
\vspace{-0.4 cm}
\begin{center}
$
\begin{array}{llll}
\delta_{FR}:&A &\rightarrow &\mathbb{N}\\
 &x &\mapsto &\delta_{FR}(x):=min \{~  |D(m_{1})| ~: ~m_{1} \geq x, ~m_{1} \in A\}.\\
\end{array}
$
\end{center}
\end{definition}
There are some well-known results about the function $\delta_{FR}$ for an arbitrary numerical semigroup $A$. One of the important results is the following
$$\delta_{FR}(x) \geq x+1-2g \text{ for all } x \in A \text{ with } x \geq c.$$
The following result gives a bound on the generalized Hamming weights of certain codes in terms of the function $\delta_{FR}$. 
\begin{theorem}{$($\cite{semi}, Theorem $46)$}
Let $A=\{0=\rho_{1}<\rho_{2}< \cdots < \rho_{n}< \cdots \}$ be an embedding dimension two numerical semigroup. Then $$d_{l}(C_{t}) \geq \delta_{FR}(t+1)+\rho_{l}$$
for $l=1, \cdots, k_{t}$, where $C_{t}$ is a code in an array of codes as in \cite{array} and $k_{t}$ is the dimension of $C_{t}$.
\end{theorem}

\section{Elementary Abelian $p$-extensions $F /\mathbb{F}_{p^{s}}$}
Let $K:=\mathbb{F}_{p^{s}}$ and $q$ be a power of $p$. Assume all the roots of the equation $T^{q}+\mu T=0$ are in $K$ (choose $s$ large enough so that all the roots are in $K$). Denote by $\beta_{1},\cdots,\beta_{q}$ the roots of $T^{q}+\mu T$ in $K$. Let $m$ be a positive integer coprime to $p$ such that $m <p^s$. Choose $m$ distinct elements $\alpha_{1},\cdots,\alpha_{m} \in K$. Let $f(x):=\prod_{i=1}^{m}(x-\alpha_{i})$.\par   Consider the function field $F /K$ defined by the equation 
\begin{equation}
y^{q} + \mu y=f(x)\in K[x],
\end{equation}
where $0 \neq \mu \in K$. \\

\begin{remark}
We have $min\{q,m\} \geq 2$, else $m=1$ and then $F/K$ is rational function field.
\end{remark}

The function field defined in equation $(3.1)$ is an elementary abelian $p$-extension of $K(x)$, as defined in \cite{book}. Some of the properties of $F/ K$ can be seen in the following lemma.

\begin{lemma}{$($\cite{book}, p.$232)$}
\vspace{-0.4 cm}
\begin{enumerate}
\item The genus of $F/K$ is $g=(q-1)(m-1)/2$.
\item The pole $P_{\infty}\in \mathbb{P}_{K(x)}$ of $x$ in $K(x)$ has a unique extension $Q_{\infty} \in \mathbb{P}_{F}$ of degree one, and $e(Q_{\infty}|P_{\infty})=q$.
\item The divisor of the differential $dx$ is 
$$(dx)=(2g-2)Q_{\infty}=((q-1)(m-1)-2)Q_{\infty}.$$
\item The pole divisor of $x$ is $(x)_{\infty}=qQ_{\infty}$ and the pole divisor of $y$ is $(y)_{\infty}=mQ_{\infty}$.
\item Let $r \geq 0$. Then, the elements $x^{i}y^{j}$ with 
$$0 \leq i, ~0\leq j \leq q-1, ~qi+mj \leq r$$
form a basis of the space $\mathcal{L}(r Q_{\infty})$ over $K$.
\item For $1 \leq i \leq m$, $1 \leq j \leq q$, $P_{\alpha_{i},\beta_{j}}$ are the places of $F/K$ of degree one.\\
\end{enumerate}
\end{lemma}

Let $Q$ be a place of a function field $F'/K'$ of degree one. An integer $l\geq 0$ is called pole number at $Q$ is there exists $z \in F'$ such that $(z)_{\infty}=lQ$. Let $p_{1}<p_{2}<\cdots$ be the sequence of pole numbers at $Q$ (that is, $p_{a}$ is the $a$-th pole number at $Q$); thus $dim (\mathcal{L}(p_{a}Q))=a$, so $p_{1}=0$. The Weierstrass semigroup $H$ at $Q$ is the set of pole numbers at $Q$. \par We have the following result regarding the  Weierstrass semigroup $H$ at $Q_{\infty}$.

\begin{lemma}{$($\cite{book}, \cite{order}$)$}
The Weierstrass semigroup $H$ at $Q_{\infty}$ is generated by $m$ and $q$ i.e.\ $H=\langle q,m \rangle$.
The largest gap number at $Q_{\infty}$ is $(2g-1)$ and $H$ is a symmetric numerical semigroup.
\end{lemma}

Let $\mathcal{X'}$ be the non-singular projective curve associated with $F/K$. Then $K(\mathcal{X'}) \cong F$. We next show that $\mathcal{X'}$ is a weak Castle curve.\\~\\
A pointed curve $(\mathcal{X}, Q)$ over $\mathbb{F}_{p^s}$ is a curve $\mathcal{X}$ defined over $\mathbb{F}_{p^s}$ together with a rational point $Q \in \mathcal{X}(\mathbb{F}_{p^s})$.
\begin{definition}{$($\cite{weakcastle1}, Definition $2.1)$}
A pointed curve $(\mathcal{X}, Q)$ over $\mathbb{F}_{p^s}$ is called weak castle if 
\vspace{-0.5 cm}
\begin{itemize}
\item the Weierstrass semigroup $H$ at $Q$ is symmetric;
\item there exist a morphism $\phi : \mathcal{X} \rightarrow \mathbb{P}^{1}$ with $div_{\infty}(\phi)=hQ$, and elements $\gamma_{1},\cdots, \gamma_{b} \in \mathbb{F}_{p^s}$ such that for all $i=1, \cdots,b$,  we have $\phi^{-1}(\gamma_{i}) \subseteq \mathcal{X}(\mathbb{F}_{p^s})$ and $\# \phi^{-1}(\gamma_{i})=h$.\\
\end{itemize} 
\end{definition}

\begin{lemma}
$(\mathcal{X'},Q_{\infty})$ is a weak Castle curve.
\end{lemma}
\begin{proof}
By Lemma $3.3$, $H$ at $Q_{\infty}$ is symmetric. Take $\phi:=x$ i.e.
\begin{center}
$\phi:\mathcal{X'} \rightarrow \mathbb{P}^{1},~P \mapsto [x(P):1].$\\
\end{center}
Then, $(x)_{\infty}=qQ_{\infty}$. For $1 \leq i \leq m$, we have $$\phi^{-1}(\alpha_{i})=\{(\alpha_{i},\beta_{1}),\cdots,(\alpha_{i},\beta_{q})\} \subseteq \mathcal{X'}(K).$$ 
(It follows from \cite{book2}, Theorem $3.1.15$ that $K$-rational points of $\mathcal{X'}$ corresponds to places of $F$ of degree one. Denote by $P_{\alpha_{i}}$ the zero of $(x-\alpha_{i})$ in $K(x)$ i.e place of degree one corresponding to $K$-rational points $\alpha_{i}$ of $\mathbb{P}^{1}$. Then, by \cite{book}( p. $232$), there are exactly $q$ places $P_{\alpha_{i},\beta_{j}}$, $1 \leq j \leq q$, of degree one lying over $P_{\alpha_{i}}$, for each $i$, $1 \leq i \leq m$. Thus, $P_{\alpha_{i},\beta_{j}}$ corresponds to $(\alpha_{i},\beta_{j})$.)\\
\end{proof}

\vspace{-1 cm}
\section{\textbf{Geometric Goppa code over $F /K $}}
In \cite{book}, Stichtenoth has investigated one-point geometric Goppa codes over the Hermitian function field. In \cite{weakcastle1}, the authors have defined one-point geometric Goppa codes over weak Castle curves (i.e.\ weak Castle codes) and studied their parameters. In this section, we define one-point geometric Goppa codes over $F/K$ as defined in \cite{weakcastle1} and determine its parameters using the ideas of \cite{book}, \cite{tru} and \cite{weakcastle1}.\par

\begin{definition}
For r $\in \mathbb{Z}$, we define\\
$$C_{r}:=C_{\mathcal{L}}(D,rQ_{\infty}),$$
where $$D:=\sum_{i=1}^{m} \sum_{j=1}^{q} P_{\alpha_{i},\beta_{j}}.$$
\end{definition}
Then, $C_{r}$ is a code of length $N:=qm$ over the field $K$. For $r <0$, $\mathcal{L}(rQ_{\infty})=\{0\}$, therefore $C_{r}=\{(0, \cdots, 0)\}$. For $r > N+(2g-2)=2qm-q-m-1$, $dim (C_{r})=N$, therefore $C_{r}=K^{N}$. So, it remains to study codes $C_{r}$ with $ 0\leq r \leq 2qm-q-m-1$. 

\subsection{Dual code of $C_{r}$}~\\~\\
In \cite{weakcastle1}, Proposition $3.1$, the authors have determined dual of weak Castle codes. In this subsection, we write a detailed proof of duality of $C_{r}$.\\

Let $f(x)$ and $D$ as in section $3$ and $4$. Let $\eta:=\frac{f'(x)}{f(x)}dx$ be a differential, then we get $\nu_{P}(\eta)=-1$ and $res_{P}(\eta)=1$ for all places $P \in supp(D)$. Therefore, we have $C_{\mathcal{L}}(D,r Q_{\infty})^{\perp}=C_{\mathcal{L}}(D,D+(\eta)-r Q_{\infty})$.\\

\begin{proposition}
The dual code of $C_{r}$ is given by $$C_{r}^{\perp}=\bar{a} \star C_{2qm-q-m-1-r},$$ where $(\bar{a})^{-1}=((f'(x))(P_{\alpha_{1},\beta_{1}}),\cdots,(f'(x))(P_{\alpha_{m},\beta_{q}})) \in (K^{*})^{N}$ and $\star$ refers to coordinate-wise product in $K^N$.
\end{proposition}
\begin{proof}
\begin{align*}
C_{r}^{\perp}&=C_{\mathcal{L}}(D, D+ (\eta)-rQ_{\infty})\\
             &=C_{\mathcal{L}}(D, D+(f'(x))-(f(x))+(dx)-rQ_{\infty})\\
             &=C_{\mathcal{L}}(D, D+(f'(x))-D+qm Q_{\infty}+(2g-2)Q_{\infty}-r Q_{\infty})\\
             &=C_{\mathcal{L}}(D, (f'(x))+(qm+(2g-2-r) Q_{\infty})\\ 
             &=\bar{a}\star C_{\mathcal{L}}(D,(2qm-q-m-1-r) Q_{\infty})\\
             &=\bar{a}\star C_{2qm-q-m-1-r}.
\end{align*}
\end{proof}

\subsection{Parameters of $C_{r}$}~\\~\\
In \cite{weakcastle1}, Proposition $3.4$, $3.6$ and $3.8$, the authors have determined dimension and minimum distance of $C_{r}$ for certain values of $r$. In this subsection, we include the known results and also determine the parameters of $C_{r}$ for other values of $r$.\\

We have the following result on the parameters of geometric Goppa code $C_{\mathcal{L}}(D',G)$.
\begin{theorem}{$($\cite{book}, $2.2.2$ and $2.2.3)$}
$C_{\mathcal{L}}(D',G)$ is an $[n,k,d]$ code with parameters $$k=dim(\mathcal{L}(G))-dim(\mathcal{L}(G-D')) \text{ and } d \geq n-deg(G).$$ If $deg(G) < n$, then $k=dim(\mathcal{L}(G))$.\\
\end{theorem}

In the next results, we determine the parameters of $C_{r}$.\\~\\
Consider the set $H$ of pole numbers at $Q_{\infty}$ (i.e.\ Weiestrass semigroup at $Q_{\infty}$). For $b \geq 0$, let $$H(b):=\{u \in H:~u \leq b\}.$$
Then, $|H(b)|=dim(\mathcal{L}(bQ_{\infty}))$. From Lemma $3.2$, we have
$$H(b)=\{u \leq b:~u=iq+jm \text{ with } i \geq 0 \text{ and } 0 \leq j \leq q-1\}.$$
Hence $$|H(b)|=|\{(i,j) \in \mathbb{N}_{0} \times \mathbb{N}_{0};~ j \leq q-1 \text{ and } iq+jm \leq b \}|.$$\\

\begin{theorem}
Suppose that $0 \leq r \leq 2qm-q-m-1$. Then the following holds:
\begin{enumerate}
\item We have $dim (C_{r})=dim(\mathcal{L}(rQ_{\infty}))-dim(\mathcal{L}((r-qm)Q_{\infty}))$. Furthermore\\
\begin{enumerate}
\item For $0 \leq r < qm$, $dim (C_{r})= |H(r)| $.
\item Define $s:=2qm-q-m-1-r$. For $qm \leq r \leq 2qm-q-m-1$, 
$$dim(C_{r})=qm-|H(s)|.$$
\item  For $qm-q-m-1 < r <qm$, $dim (C_{r})=r+1-\frac{(q-1)(m-1)}{2}$.
\end{enumerate}
\vspace{0.2 cm}
\item The minimum distance $d(C_{r})$ of $C_{r}$ satisfies $$d(C_{r}) \geq qm-r.$$
 If $r=qb$, where $0 \leq b < m $ or if $r=cm$, where $0 \leq c <q $, then $d(C_{r})=qm-r$. In addition, if $r \geq qm-q-m$ then $C_{r}$ is not an MDS code.
\end{enumerate}
\end{theorem}
\begin{proof}
\begin{enumerate}
\item As $D \sim qmQ_{\infty}$, from Theorem $4.3$ we have, $$dim (C_{r})=dim(\mathcal{L}(rQ_{\infty}))-dim(\mathcal{L}((r-qm)Q_{\infty}))$$ and if $0 \leq r < qm$, then 
 $dim (C_{r})= dim(\mathcal{L} (r Q_{\infty}))= |H(r)|.$\\~\\
 For $qm \leq r \leq 2qm-q-m-1$, we have $0 \leq s <qm$. So, it follows from Proposition $4.2$
 $$dim~C_{r}=qm-dim~C_{r}^{\perp}=qm-dim~C_{s}=qm-|H(s)|.$$
 If $qm-q-m-1 < r <qm$ (i.e.\ $2g-2< deg(rQ_{\infty}) < N$), then Riemann-Roch theorem yields 
 \vspace{-0.2 cm} $$dim (C_{r})=dim(\mathcal{L}(rQ_{\infty}))=deg(rQ_{\infty})+1-g=r+1-\frac{(q-1)(m-1)}{2}.$$
\item The inequality $d(C_{r}) \geq qm-r$ directly follows from Theorem $4.3$. If $r=qb$, where $0 \leq b < m $, choose $b$ distinct elements from the set $\{\alpha_{1},\cdots, \alpha_{m}\}$. Let us call these elements $\gamma_{1},\cdots,\gamma_{b}$. Then the element
$$z_{1}:=\prod_{j=1}^{b} (x-\gamma_{j}) \in \mathcal{L}(rQ_{\infty})$$
has exactly $qb=r$ distinct zeros in $D$. The weight of the corresponding codeword in $C_{r}$ is $qm-r$. Hence, $d(C_{r})=qm-r$.\\
~\\
Similarly, if $r=cm$, where $0 \leq c < q $, choose $c$ distinct elements from the set $\{\beta_{1},\cdots, \beta_{q}\}$. Let us call these elements $\tau_{1},\cdots,\tau_{c}$. Then the element
$$z_{2}:=\prod_{j=1}^{c} (y-\tau_{j}) \in \mathcal{L}(rQ_{\infty})$$
has exactly $cm=r$ distinct zeros in $D$. The weight of the corresponding codeword in $C_{r}$ is $qm-r$. Hence, $d(C_{r})=qm-r$.\\
~\\
If $r=qb$ and $C_{r}$ is an MDS code, then $d(C_{r})=qm-dim(C_{r})+1$ implies $g=0$, which is not possible. Similarly, for $r=cm$.
\end{enumerate}
\end{proof}
\vspace{-0.5 cm}
In the following theorem, we determine the minimum distance of $C_{r}$ for $qm \leq r \leq 2qm-q-m-1$. Using the ideas from \cite{tru} and Theorem $2.2$, we have the following result.
\begin{theorem}
Assume $m>q$. For $qm \leq r \leq 2qm-q-m-1$ we have $0 \leq r^{\perp}:=2qm-q-m-1-r \leq qm-q-m-1$. Let $t^{\perp} \leq r^{\perp}$ be the largest integer such that $t^{\perp}$ is a pole number at $Q_{\infty}$ i.e.\ $t^{\perp}=aq+bm$ where $0 \leq a \leq m-2$ and $0 \leq b \leq q-1$. Then, the minimum distance of $C_{r}$ satisfies
$$d(C_{r})=a+2.$$
\end{theorem}
\vspace{-0.7 cm}
\begin{proof}
Let $\mathbf{H}$ be a parity check matrix of $C_{r}$. From Lemma $3.2$, we have $\{1,x,y,\cdots,x^{a}$, $x^{a-1}y, \cdots, y^{b}\}$ is a basis for $\mathcal{L}(t^{\perp}Q_{\infty})$. Choose $\beta \in K$ such that $\beta^{q}+ \mu \beta=0$. Let $\mathbf{H}_{1}$ be a submatrix  of $\mathbf{H}$ with columns corresponding to $P_{\alpha_{1},\beta},\cdots,P_{\alpha_{a+2},\beta}$. We write $\mathbf{H}_{1}$ in the following form using row reduction.
\[
\mathbf{H}_{1}=
  \begin{bmatrix}
    1 & 1 & 1 & \cdots & 1 \\
    \alpha_{1} & \alpha_{2} & \alpha_{3} & \cdots & \alpha_{a+2}\\
    \alpha_{1}^{2} & \alpha_{2}^{2} & \alpha_{3}^{2} & \cdots & \alpha_{a+2}^{2}\\
    \vdots & \vdots & \vdots & ~ & \vdots \\
    \alpha_{1}^{a} & \alpha_{2}^{a} & \alpha_{3}^{a} & \cdots & \alpha_{a+2}^{a}\\
     0 & 0 & 0 & \cdots & 0 \\
     \vdots & \vdots & \vdots & ~ & \vdots \\
     0 & 0 & 0 & \cdots & 0 
  \end{bmatrix}
\]
Here, $rank(\mathbf{H}_{1})=a+1$ and $\mathbf{H}_{1}$ has $a+2$ columns, so the columns of $\mathbf{H}_{1}$ are linearly dependent. Therefore, $d(C_{r}) \leq a+2$.

On the other hand, we choose any $a+1$ distinct columns from $\mathbf{H}$. Let us call this matrix $\mathbf{H}_{2}$. Since each column of $\mathbf{H}$ corresponds to a place $P_{\alpha,\beta}$ of degree one, we reorder columns of $\mathbf{H}_{2}$ according to $\alpha's$ as follows.
\[
  \begin{matrix}
    P_{\alpha_{1},\beta_{1,1}}, & P_{\alpha_{1},\beta_{1,2}}, & \cdots, & P_{\alpha_{1},\beta_{1,w_{1}}}\\
    P_{\alpha_{2},\beta_{2,1}}, & P_{\alpha_{2},\beta_{2,2}}, & \cdots, & P_{\alpha_{1},\beta_{2,w_{2}}}\\
    \vdots & \vdots & \vdots & \vdots \\
    P_{\alpha_{\gamma},\beta_{\gamma,1}} & P_{\alpha_{\gamma},\beta_{\gamma,2}} & \cdots, & P_{\alpha_{\gamma},\beta_{\gamma,w_{\gamma}}}\\
  \end{matrix}
\]
where $\alpha_{i}$'s are pairwise distinct and $w_{1}+w_{2}+\cdots+w_{\gamma}=a+1$ with $w_{1} \geq w_{2} \geq \cdots \geq w_{\gamma} \geq 1$. For $0 \leq j_{i} \leq w_{i}-1; ~1 \leq i \leq \gamma$, $x^{i-1}y^{j_{i}}$ belongs to basis of $\mathcal{L}(t^{\perp}Q_{\infty})$.
We rewrite these basis elements in the form
\[
\begin{matrix}
    1, & y, & y^{2}, & \cdots, & y^{w_{1}-1}\\
    x, & xy, & xy^{2}, & \cdots, & xy^{w_{2}-1}\\
    x^{2}, & x^{2}y, & x^{2}y^{2}, & \cdots, & x^{2}y^{w_{3}-1}\\
    \vdots & \vdots & \vdots & ~ & \vdots \\
    x^{\gamma-1}, & x^{\gamma-1}y, & x^{\gamma-1}y^{2}, & \cdots, & x^{\gamma-1}y^{w_{\gamma}-1}\\
  \end{matrix}
\]
  
Then, we extract an $(a+1) \times (a+1)$ submatrix $\mathbf{H'}$ of $\mathbf{H}_{2}$ such that each row corresponds to a function above in the given order. That is, $\mathbf{H'}=[\mathbf{H'}_{i,j}]$, $i,j=1,2,\cdots,\gamma$ where $\mathbf{H'}_{i,j}$ is a $(w_{i} \times w_{j})$ matrix with $\mathbf{H'}_{i,j}=\alpha_{j}^{i-1}\mathbf{B}_{i,j}$ with
\[
\mathbf{B}_{i,j}=
\begin{bmatrix}
    1 & 1 & 1 & \cdots & 1\\
    \beta_{j,1} & \beta_{j,2} & \beta_{j,3} & \cdots & \beta_{j,w_{j}}\\
    \beta_{j,1}^{2} & \beta_{j,2}^{2} & \beta_{j,3}^{2} & \cdots & \beta_{j,w_{j}}^{2}\\
    \vdots & \vdots & \vdots & ~ & \vdots \\
    \beta_{j,1}^{w_{i}-1} & \beta_{j,2}^{w_{i}-1} & \beta_{j,3}^{w_{i}-1} & \cdots & \beta_{j,w_{j}}^{w_{i}-1}\\
  \end{bmatrix}
\]
Then, from \cite{tru}, Lemma $2$ and Lemma $3$, $$det(\mathbf{H'})=(\prod_{i=1}^{\gamma} det(\mathbf{B}_{i,i})).(\prod_{j=2}^{\gamma} \rho_{j}^{w_{j}})$$
where $$ \rho_{j}=\prod_{i=1}^{j-1}(\alpha_{j}-\alpha_{i}),~j=2,3,\cdots,\gamma.$$
And any $a+1$ columns of $\mathbf{H}$ are linearly independent over $K$. Hence, $d(C_{r}) \geq a+2$.
\end{proof}

\section{\textbf{Generalized Hamming weights of code $C_{r}$}}
In \cite{weakcastle1}, Proposition $3.7$ and $3.8$, the authors have determined bounds on the generalized Hamming weights of $C_{r}$ using the concepts of gonality and order bounds. But in general, the computation of gonality is a difficult task. In this section, we determine the exact values of generalized Hamming weights, in particular, the second generalized Hamming weight of $C_{r}$ in a few cases.\\

Following the ideas of \cite{ghwag}, we have the following lemma.
\begin{lemma}
Let $r \leq qm$ be a pole number at $Q_{\infty}$. Then, $r=iq+jm$ where $i \geq 0$ and $0 \leq j \leq q-1$. If either $i=0$ or $j=0$ then $\exists$ a divisor $0 \leq D' \leq D$ such that $rQ_{\infty} \sim D'$.
\end{lemma}
\begin{proof}
 For $i=0$ and $j=0$, $D'=0$ works. Now if $i=0$ and $j \neq 0$, then $r=jm$. Choose $j$ elements from $\beta_{1}, \cdots, \beta_{q}$. Denote these elements by $\tau_{1},\cdots, \tau_{j}$. Define $$z:=\prod_{t=1}^{j}(y-\tau_{t}).$$ Then, $(z)=D'-rQ_{\infty}$ which implies $D' \sim rQ_{\infty}$. 
The lemma can be proved similarly for $j=0$ and $i \neq 0$.\\
\end{proof}

\begin{definition}
A positive integer $r \leq qm$ is said to have property (*) if $r$ is a pole number at $Q_{\infty}$, $r=iq+jm$ for $i \geq 0$, $0 \leq j \leq q-1$ and either $i=0$ or $j=0$.\\
\end{definition} 
 
\begin{theorem}
If for $1 \leq l \leq dim(C_{r})$, $r-p_{l}$ or $qm-r+p_{l}$ has the property (*), then $$d_{l}(C_{r}) \leq qm-r+p_{l}.$$
\end{theorem}
\vspace{-0.5 cm}
\begin{proof}
If $r-p_{l}$ has the property (*), then from Lemma $5.1$ we have a divisor $D'$ such that $0 \leq D' \leq D$ and $(r-p_{l})Q_{\infty} \sim D'$. So, $$dim(\mathcal{L}(rQ_{\infty}-D'))=dim( \mathcal{L}(p_{l}Q_{\infty}))=l.$$ Thus, from Theorem $2.3$ it follows that $$d_{l}(C_{r}) \leq qm-r+p_{l}.$$\par
 Now if $qm-r+p_{l}$ has the property (*), then again there exists a divisor $D''$ such that $0 \leq D'' \leq D$ and $(qm-r+p_{l})Q_{\infty} \sim D''$.  Also, $qmQ_{\infty} \sim D$, which implies $D-rQ_{\infty}+p_{l}Q_{\infty} \sim D''$. Therefore, $D':=D-D'' \sim (r-p_{l})Q_{\infty}$. Hence,  $$d_{l}(C_{r}) \leq qm-r+p_{l}.$$
\end{proof}
  
An immediate corollary to Theorem $5.3$ is the following.  
\begin{corollary} 
If  for $1 \leq r <qm $, $r$ or $qm-r$ has the property (*), then the minimum distance of $C_{r}$ is  $d(C_{r})=qm-r$.
\end{corollary}

\begin{remark}
 We have $d_{k}(C_{r})=qm$ where $k=dim(C_{r})$.
\end{remark}

From Lemma $3.3$, the Weierstrass semigroup $H=\langle q,m \rangle$ at $Q_{\infty}$ is an embedding dimension two numerical semigroup. Thus, from Theorem $2.5$ we get the following results on the second generalized Hamming weight of $C_{r}$.
\begin{theorem}
Assume that $m >q$. For $r <qm$, $d_{2}(C_{r}) \geq qm-r+q.$\\
\end{theorem}

\begin{theorem}
Assume that $m >q$. For $r <qm-1$, if $r-q$ or $qm-r+q$ satisfies the property (*), then $$d_{2}(C_{r})=qm-r+q.$$
\end{theorem}
\begin{proof}
Applying Theorem $5.3$, we get $d_{2}(C_{r})\leq qm-r+q.$\\~\\
On the other hand, as $\rho_{2}=q$ and as the dual code of $C_{r}$ form an array of codes ( for details see \cite{array}), from Theorem $2.5$ and Proposition $4.2$, we have 
\begin{align*}
d_{2}(C_{r})&=d_{2}(C_{2qm-q-m-1-r}^{\perp})\\
            &\geq \delta_{FR}(2qm-q-m-r)+q.
\end{align*}
Since $r<qm$, we have $2qm-q-m-r \geq 2g=qm-q-m+1$, therefore
\begin{align*}
d_{2}(C_{r})&\geq \delta_{FR}(2qm-q-m-r)+q\\
&\geq 2qm-q-m-r+1-(qm-q-m+1)+q\\
&=qm-r+q.
\end{align*}
Hence proved.
\end{proof}

Theorem $5.7$ can be generalised for all $l$, $1 \leq l \leq dim(C_{r})$.
\begin{theorem}
For $r <qm-1$ and $1 \leq l \leq dim(C_{r})$, if $r-p_{l}$ or $qm-r+p_{l}$ satisfies the property (*), then $$d_{l}(C_{r})=qm-r+p_{l}.$$
\end{theorem}
\begin{proof}
From Theorem $5.3$ we have $$d_{l}(C_{r}) \leq qm-r+p_{l}.$$ 
The reverse inequality follows from the proof of Theorem $5.7$.
\end{proof}

The following result is stated in Munuera \cite{ghwag}.
\begin{theorem}{$($\cite{ghwag}, Proposition $4)$}
Let $C_{\mathcal{L}}(D',G)$ be a code of dimension $k$ and abundance $dim(\mathcal{L}(G-D'))=:a \geq 0$. If there is a place $Q$ of degree one not in $D'$ and $C_{\mathcal{L}}(D',G-p_{l+a}Q) \neq \{0\}$, where $p_{l}$ is $l$-th pole number at $Q$, then for every $l$, $1 \leq l \leq k$,
$$d_{l}(C_{\mathcal{L}}(D',G)) \leq d_{1}(C_{\mathcal{L}}(D',G-p_{l+a}Q)).$$
\end{theorem}

Using Theorem $5.6$ and Theorem $5.9$, we get the following result.
\begin{theorem}
Assume $m >q$. For $r< qm$. If $r-q$ and $qm-r+q$ doesn't satisfy the property (*), then $$qm-r+q \leq d_{2}(C_{r}) \leq qm-\overline{(r-q)}.$$
where $\overline{(r-q)}$ is the largest pole number less than or equal $(r-q)$ that satisfies the property (*).
\end{theorem}

\begin{theorem}
For $1 \leq l \leq dim~C_{r}$ and $r< qm$. If $r-p_{l}$ and $qm-r+p_{l}$ doesn't satisfy the property (*), then $$d_{l}(C_{r}) \leq qm-\overline{(r-p_{l})}.$$
where $\overline{(r-p_{l})}$ is the largest pole number less than or equal $(r-p_{l})$ that satisfies the property (*).
\end{theorem}

In the following results, we determine the second generalized Hamming weight of $C_{r}$ when $qm \leq r$. To prove the results, we use the following result stated in Munuera \cite{ghwag}.
\begin{theorem}{$($\cite{ghwag}, Proposition $6)$}
Let $C=C_{\mathcal{L}}(D',G)$ be a code of dimension $k$ and  $dim(\mathcal{L}(G-D'))=:a > 0$. Then, for $1 \leq l \leq k$, we have $d_{l}(C) \leq deg(D'')$ for every effective divisor $D'' \leq D$ such that $dim(\mathcal{L}(D''))>l$.\\
\end{theorem}

\begin{theorem}
Assume that $m >q$. If $qm \leq r \leq 2qm-q-m-1$, then $$d_{2}(C_{r}) \leq \text{min} \{2q, m\}.$$
\end{theorem}
\vspace{-0.6 cm}
\begin{proof}
Since $p_{3}=\text{min} \{2q, m\}$ therefore, from Lemma $5.1$, $\exists$ a divisor $D'$ such that $0\leq D' \leq D$ and $p_{3}Q_{\infty} \sim D'$. Then, $dim(\mathcal{L}(D'))=dim(\mathcal{L}(p_{3}Q_{\infty}))=3$. Hence, $d_{2}(C_{r}) \leq deg~D'=p_{3}=\text{min} \{2q, m\}$.\\
\end{proof}

\begin{theorem}
Assume that $m>2q$. Suppose $qm \leq r \leq 2qm-q-m-1$. Then $0 \leq r^{\perp}:=2qm-q-m-1-r \leq qm-q-m-1$. Let $t^{\perp} \leq r^{\perp}$ be the largest integer such that $t^{\perp}$ is a pole number at $Q_{\infty}$ i.e.\ $t^{\perp}=aq+bm$ where $0 \leq a \leq m-3$ and $0 \leq b \leq q-1$. Then,
$$d_{2}(C_{r})=a+3.$$
\end{theorem}
\vspace{-0.7 cm}
\begin{proof}
Let $\mathbf{H}$ be a parity check matrix for $C_{r}$ over $K$. Choose $\beta \in K$ such that $\beta^{q}+\mu \beta=0$. $\{1,x,y,\cdots,x^{a},x^{a-1}y, \cdots, y^{b}\}$ is a basis for $\mathcal{L}(t^{\perp}Q_{\infty})$. Let $\mathbf{H}_{1}$ be a submatrix of $\mathbf{H}$ with columns corresponding to $P_{\alpha_{1},\beta},\cdots,P_{\alpha_{a+3},\beta}$ (possible since $a+3\leq m$). By using row reduction, we make $\mathbf{H}_{1}$ as follows.
\[
\mathbf{H}_{1}=
  \begin{bmatrix}
    1 & 1 & 1 & \cdots & 1 \\
    \alpha_{1} & \alpha_{2} & \alpha_{3} & \cdots & \alpha_{a+3}\\
    \alpha_{1}^{2} & \alpha_{2}^{2} & \alpha_{3}^{2} & \cdots & \alpha_{a+3}^{2}\\
    \vdots & \vdots & \vdots & ~ & \vdots \\
    \alpha_{1}^{a} & \alpha_{2}^{a} & \alpha_{3}^{a} & \cdots & \alpha_{a+3}^{a}\\
     0 & 0 & 0 & \cdots & 0 \\
     \vdots & \vdots & \vdots & ~ & \vdots \\
     0 & 0 & 0 & \cdots & 0 
  \end{bmatrix}
\]
Here, $rank(\mathbf{H}_{1})=a+1$ and $\mathbf{H}_{1}$ has $a+3$ columns. So, by Theorem $2.2$, $d_{2}(C_{r}) \leq a+3$.\\
~\\
On the other hand from Theorem $2.1$, $d_{2}(C_{r}) \geq d_{1}(C_{r})+1=(a+2)+1=a+3$.\\
Hence proved.
\end{proof}

\textbf{Example:} For $p=2$, $K=\mathbb{F}_{4}$, consider the function field $F=K(x,y)$ defined by $$y^{2}+\omega y=x(x-1)(x-\omega),$$ where $\omega$ is a primitive element of $K$. Here, we have $q=2$, $m=3$ and genus $g=1$. The list of values of length, dimension, minimum distance and second generalized Hamming weight of $C_{r}$ is given by the following table.
 \begin{center}
\begin{tabular}{ |l|l|l|l|l|} 
 \hline
  $r$ & $N$ & $dim(C_{r})$ & $d(C_{r})$ & $d_{2}(C_{r})$\\ 
 \hline
 $1$ & $6$ & $1$ & $6$ &$-$ \\ 
 $2$ & $6$ & $2$ & $4$ &$6$ \\ 
 $3$ & $6$ & $3$ & $3$ &$\geq 5$ and $\leq 6$\\ 
 $4$ & $6$ & $4$ & $2$ &$4$ \\ 
 $5$ & $6$ & $5$ & $\geq 1$ & $3$ \\ 
 $6$ & $6$ & $5$ & $2$ & $\leq 4$ \\ 
 \hline
\end{tabular}
\end{center}

\section{\textbf{Condition for quasi-self-duality and self-duality of codes}}
A linear code $C$ is called self-dual if $C=C^{\perp}$, where $C^{\perp}$ is the dual of $C$ with respect to Euclidean scalar product on $\mathbb{F}_{p^{s}}^{n}$. Self-dual codes are an important class of linear codes. In this section, we give a simple criterion for the self-duality of geometric Goppa codes over $F/K$.
 
We have the following result from \cite{equal}. But first, we introduce a definition for an arbitrary algebraic function field $F'/K'$ of genus $g'$.\\

\begin{definition} $($\cite{equal}, Definition $2.15)$
Choose $n$ places $P_{1},\cdots,P_{n}$ of degree one of $F'$ and $D':=P_{1}+\cdots+P_{n}$. We call two divisors $G$ and $H$ equivalent with respect to $D'$ if there exists $u \in F'$ such that $H=G+(u)$ and $u(P_{i})=1$, for all $i=1,\cdots,n.$\\
\end{definition}

\begin{proposition}$($\cite{equal}, Corollary $4.15)$
Suppose $n > 2g' + 2$. Let $G$ and $H$ be two divisors of the same degree $m'$ on $F'$. If $C_{\mathcal{L}}(D',G)$ is not equal to $0$ nor to $(K')^{n}$ and $2g' - 1 < m' < n -1$, then $C_{\mathcal{L}}(D',G)=C_{\mathcal{L}}(D',H)$ if and only if $G$ and $H$ are equivalent with respect to $D'$.\\
\end{proposition}

\begin{theorem}
If $qm-q-m+1 \leq r \leq qm-2$ then $C_{r}$ is quasi-self-dual if and only if $r=(2qm-q-m-1)/2$.
\end{theorem}
\begin{proof}
If  $qm-q-m+1 \leq r \leq qm-2$, then from Proposition $6.2$, $C_{r}$ is quasi-self-dual if and only if $r=(2qm-q-m-1)/2$.
\end{proof}

Let $G$ be a divisor of $F$ with $deg(G)=\frac{2qm-q-m-1}{2}$. Clearly, $qm >qm-q-m+3 \text{~~and~~} qm-q-m < deg(G) < qm-1$. Let $H:=D+(\eta)-G$ with $D$ as in section $4$. Then, $deg(G)=deg(H)$. In the following theorem, we give a condition for the self-duality of code $C_{\mathcal{L}}(D,G)$.\\ 
\begin{theorem}
$C_{\mathcal{L}}(D,G)$ is self-dual if and only if $2G$ is equivalent to $(f'(x))+(2qm-q-m-1)Q_{\infty}$ with respect to $D$. 
\end{theorem}
\vspace{-0.5 cm}
\begin{proof}
By Proposition $6.2$,
\begin{align*}
& C_{\mathcal{L}}(D,G)=C_{\mathcal{L}}(D,D+(\eta)-G)\\
\Leftrightarrow &~G=D+(\eta)-G+(u) \text{ for some } u \in F \text{ such that } u(P)=1 \text{ for each place } P \in supp(D)\\
\Leftrightarrow &~(u)+(\eta)=2G-D\\
\Leftrightarrow &~(u)+(f'(x))+(dx)-(f(x))=2G-D\\
\Leftrightarrow &~(u)+(f'(x))+[(q-1)(m-1)-2]Q_{\infty}-D+ qm Q_{\infty}=2G-D\\
\Leftrightarrow &~(f'(x))=2G-(2qm-q-m-1)Q_{\infty}-(u).
 \end{align*} 
\end{proof}

\begin{example}
Let $p=2$. Let $K=\mathbb{F}_{4}$. Let $\omega$ be a primitive element of $\mathbb{F}_{4}$. Consider $F=K(x,y)$ with 
$$y^{2}+y=x(x-1)(x-\omega).$$
Therefore, all roots of $T^{2}+T$ is in $K$. The genus of $F/K$ is $g=1$. Let $$f(x):=x(x-1)(x-\omega).$$ Let $P_{0},P_{1}$ and $P_{\omega}$ denote zero in $K(x)$ of $x,(x-1)$ and $(x-\omega)$ respectively. Then, each of $P_{0},P_{1}$ and $P_{\omega}$ has exactly two extensions in $F$. Similarly, the zero of $(x-\omega^{2})$ denoted by $P_{\omega^{2}}$ has two extensions in $F$, say, $Q_{1}$ and $Q_{2}$. Let $D:=(f(x))_{0}$ and let $G$ be a divisor of $F$ equivalent to $Q_{1}+Q_{2}+Q_{\infty}$ with respect to $D$. Then, $C_{\mathcal{L}}(D,G)$ is self-dual. \\
Conversely, if $C_{\mathcal{L}}(D,G)$ is self-dual code then $2G$ is equivalent to $2(Q_{1}+Q_{2}+Q_{\infty})$ with respect to $D$.
\end{example}

\section{\textbf{Quantum codes from one-point geometric Goppa codes on $F/ \mathbb{F}_{p^s}$}}
In this section, we construct quantum codes from one-point geometric Goppa codes on $F/ \mathbb{F}_{p^s}$. First, we give a brief introduction to quantum codes.\par
Let $q_{0}$ be a prime power. Let $V_{n}:=(\mathbb{C}^{q_{0}})^{\otimes n}$ denotes the $n$-th tensor power of $q_{0}$-dimensional Hilbert space $\mathbb{C}^{q_{0}}$. An $[[n, k, d]]_{q_{0}}$ quantum code is a $q_{0}^{k}$-dimensional vector subspace of $V_{n}$ with minimum distance $d$. The connection between quantum codes and classical linear codes was established by Calderbank et al. \cite{quantum2}. Since then, many classes of quantum codes have been constructed by using classical error-correcting codes.\par The Singleton bound for quantum codes states that an $[[n, k, d]]_{q_{0}}$ quantum code must obey $2d \leq n-k+2.$ The quantum Singleton defect is defined as $\delta^{Q}:=n-k-2d+2 \geq 0$, and the relative quantum Singleton defect is $\Delta^{Q}:=\delta^{Q}/n$. If $\delta^{Q}=0$, then the code is said to be quantum MDS.\\

The following lemma gives a construction of quantum codes from classical linear codes.
\begin{lemma} {$($\cite{quantum}, Lemma $17~ (a)$ $)$}
Let $C_{1}$ and $C_{2}$ denote two linear codes with parameters $[n, k_{i}, d_{i}]_{q_{0}}$, $i = 1, 2$, and assume that $C_{1} \subset C_{2}$. Then there exists an $[[n, k_{2}-k_{1},d]]_{q_{0}}$ code with $d = min\{wt(c) ~:~ c \in (C_{2} \backslash C_{1}) \cup (C^{\perp}_{1} \backslash C^{\perp}_{2} )\}$, where $wt(c)$ is the weight of $c$.
\end{lemma}

We apply Lemma $7.1$ to obtain quantum codes from one-point geometric Goppa codes constructed in section $4$.
\begin{proposition}
Let $m$ and $q ~(=p^{k'})$ be as in section $3$. Let $a,b$ be positive integers such that $qm-q-m-1<a<b<qm$. Then there exists a $[[qm,b-a ,d]]_{p^s}$ code with $d \geq min \{~qm-b,a-qm+q+m+1\}$.\par In addition, if $m=q+x$ for some $x \in \mathbb{N}$ and if $a$ is of the form $cq$ for a positive integer $c$ such that $qm-b \leq a-qm+q+m+1$, then the relative quantum singleton defect $\Delta_{k'}^{Q}$ of the obtained quantum code satisfies $$lim_{k' \rightarrow \infty}~ \Delta_{k'}^{Q}=0.$$
\end{proposition}
\begin{proof}
Let $C_{1}:=C_{a}$ and $C_{2}:=C_{b}$ as in section $4$. Then, $C_{1}$ is a $[qm,a+1-\frac{(q-1)(m-1)}{2},d_{1}\geq qm-a]_{p^s}$ code and $C_{2}$ is a $[qm,b+1-\frac{(q-1)(m-1)}{2},d_{2}\geq qm-b]_{p^s}$ code. Also $C_{1} \subset C_{2}$. It follows from Lemma $7.1$ that there exists a $[[qm,b-a ,d]]_{p^s}$ code with $d \geq min \{~qm-b,a-qm+q+m+1\}$.\\

Now if $m=q+x$ and $a=cq$ such that $qm-b \leq a-qm+q+m+1$, then
\begin{align*}
\Delta_{k'}^{Q}&=\frac{qm-b+a-2d+2}{qm}\\
&\leq \frac{qm-b+a-2qm+2b+2}{qm}\\
&\leq \frac{a+1}{qm}
=\frac{cq+1}{q(q+x)}
=\frac{c p^{k'}+1}{p^{k'}(p^{k'}+x)}
\rightarrow 0 \text{ as } k'\rightarrow \infty.
\end{align*}
\end{proof}

\begin{example}
Let $K=\mathbb{F}_{4}$. Let $q=2,~m=3,~a=2,~b=4$. Then $a$ and $b$ satisfy the conditions of Proposition $7.2$ and we get a $[[6,2, \geq 2]]_{4}$ code and its parameters are the best possible according to table in \cite{table}. Similarly, we get quantum codes with parameters $[[6,4,\geq 1]]_{4}$, $[[10,4,\geq 2]]_{8}$, $[[12,4,\geq 2]]_{9}$, $[[14,6,\geq 2]]_{16}$, etc.
\end{example}

\section{\textbf{Convolutional codes from one-point geometric Goppa codes on $F/ \mathbb{F}_{p^s}$}}

Consider the polynomial ring $R = \mathbb{F}_{q_{0}}[X]$. A convolutional code $C$ is an $R$-submodule of rank $k$ of the module $R^n$. Let $G(X)=(g_{i j}(X)) \in (\mathbb{F}_{q_{0}} [X])^{k \times n}$ be a generator matrix of $C$ over $\mathbb{F}_{q_{0}}[X]$, $\gamma_{i} = max \{deg ~g_{i j}(X)~: ~1 \leq j \leq  n\}$, $\gamma = \sum_{i=1}^{k} \gamma_{i}$, $m' = max \{ \gamma_{i}~:~ 1 \leq i \leq k \}$, and $d_{f}$ be the minimum weight of $c \in C$. Then we say that $C$ has length $n$, dimension $k$, degree $\gamma$, memory $m'$, and free distance $d_{f}$. If $m' = 1$, $C$ is said to be a unit-memory convolutional code and is denoted by $(n, k, \gamma ; 1, d_{f})_{q_{0}}$.\par
The following theorem describes a method to construct convolutional codes from geometric Goppa codes.
\begin{lemma} { $($\cite{convolutional}, Theorem $3$ $)$}
Let $F'/\mathbb{F}_{q_{0}}$ be a function field of genus $g'$. Consider the code $C_{\Omega}(D',G)$ with $2g'-2<deg(G)<n$, where $deg(G)$ is the degree of the divisor $G$. Then there exists a unit-memory convolutional codes with parameters $(n,k-l,l;1,d_{f} \geq d)_{q_{0}}$, where $l \leq k/2$, $k=deg(G)+1-g'$ and $d \geq n-deg(G)$.
\end{lemma}

In this following proposition, we construct unit-memory convolutional codes from one-point geometric Goppa code on $F/\mathbb{F}_{p^s}$ using Lemma $8.1$. The length of convolutional codes so obtained is $qm$ where $q=p^{k'}$ and $gcd(m,q)=1$ (such that $max\{q,m\}<p^s)$ and is different from codes in \cite{convo2}.
\begin{proposition}
Let $qm-q-m-1<r<qm$ where $r$ satisfies the property $(*)$ (Definition $5.2$). Then there exists a unit-memory convolutional codes with parameters $(qm,r+1-g-a,a;1,d_{f} \geq d)_{p^s}$, where $ g=(q-1)(m-1)/2$, $a \leq (r+1-g)/2$  and $d=qm-r$.
\end{proposition}
\begin{proof}
We consider the code $C_{\Omega}(D,r Q_{\infty})$ on $F/ \mathbb{F}_{p^s}$ with $D$ as in section $4$. Since $qm-q-m-1<r<qm$, from Lemma $8.1$ we get a unit-memory convolutional codes with parameters $(qm,r+1-g-a,a;1,d_{f} \geq d)_{p^s}$, where $g=(q-1)(m-1)/2$, $a \leq (r+1-g)/2$. As $r$ satisfies the property $(*)$, $d=qm-r$ from Corollary $5.4$.
\end{proof}

\begin{example}
Let $K=\mathbb{F}_{4}$, $q=2,~m=3,~r=4$. Then we get a unit-memory convolutional code with parameters $(6,3,1;1, \geq 2)_{4}$.\\
Similarly, we get unit-memory convolutional codes with parameters $(10,3,2;1,\geq 4)_{8}$, $(14,6,2;1, \geq 4)_{16}$, etc.
\end{example}

\section{\textbf{Locally recoverable codes from $F/ \mathbb{F}_{p^s}$}}
A code $C \subset \mathbb{F}_{q_{0}}^{n}$ is LRC with locality $r_{0}$ if for every $i \in [n]:= \{1, 2,\cdots , n \}$ there exists a subset $A_{i} \subset [n] \backslash \{i\}$, $|A_{i}| \leq r_{0}$ and a function $\phi_{i}$ such that for every codeword $x \in C$ we have $x_{i} =\phi_{i}(\{ x_{j} , j \in A_{i}\})$. An LRC code $C$ of length $n$, cardinality $q_{0}^k$ and locality $r_{0}$ is denoted by $(n,k,r_{0})$. \\The minimum distance of an $(n,k,r_{0})$ LRC code satisfies the inequality
\begin{equation} d \leq n-k- \left \lceil \frac{k}{r_{0}} \right \rceil +2. 
\end{equation}
The codes for which equality holds in equation $(9.1)$ are called optimal LRC codes. The rate of an $(n,k,r_{0})$ LRC code satisfies the inequality
\begin{equation} \frac{k}{n} \leq \frac{r_{0}}{r_{0}+1}. 
\end{equation}

In \cite{lrc}, the authors have constructed LRC codes on algebraic curves. We describe this construction briefly in this section. For more details see \cite{lrc}.\\
Let $F_{X} / \mathbb{F}_{q_{0}}$ and $F_{Y}/ \mathbb{F}_{q_{0}}$ be function fields. Let $\psi : X \rightarrow Y$ be a rational separable map of degree $r_{0} + 1$ between smooth projective absolutely irreducible curves $X$ and $Y$ corresponding to $F_{X}$ and $F_{Y}$, respectively. 
Let $\psi^{*}: F_{Y} \rightarrow F_{X}$ be the corresponding map of function fields. Since $\psi$ is separable, the primitive element theorem implies that there exists a function $y \in F_{X}$ such that $F_{X} = F_{Y}(y)$. The function $y$ can be considered as a map $y : X \rightarrow \mathbb{P}^{1}$, and we denote its degree $deg(y)$ by $h$. Let $S = \{P_{1},\cdots, P_{s_{0}}\}$ be a set of places of degree one in $F_{Y}$ and $A$ be a positive divisor of degree $l \geq 1$ whose support is disjoint from $S$. For each $j$, let $\{P_{ij}\}$ be the collection of places of $F_{X}$ over $P_{j}$. We assume that each $P_{j}$ splits completely in $F_{X}$. Let $\{ f_{1},\cdots, f_{m_{0}} \}$ be a basis of the linear space $\mathcal{L}(A)$. Let $V$ be the subspace of $F_{X}$ of dimension $r_{0}m_{0}$ generated by the functions $\{f_{j}y^{i},~ i = 0,\cdots, r_{0}-1, ~j = 1,\cdots, m_{0}\}$. Then, the code $C(A, \psi)$ is defined as the image of the map
$$e:V \rightarrow \mathbb{F}_{q_{0}}^{(r_{0}+1)s_{0}}$$
$$F \mapsto (F(P_{ij}), ~i = 0, \cdots,r_{0},~ j = 1,\cdots, s_{0}).$$

\begin{theorem}{ $($\cite{lrc},Theorem $3.1$ $)$}
The subspace $C(A, \psi)\subseteq \mathbb{F}_{q_{0}}^{(r_{0}+1)s_{0}}$ forms an $(n, k, r_{0})$ linear LRC code with the parameters
$$n = (r_{0} + 1)s_{0},$$
$$k = r_{0}m_{0} \geq r_{0}(l-g_{Y} + 1),$$
$$d \geq n-(r_{0}+1)l-(r_{0}-1)h,$$
provided that the right-hand side of the inequality for $d$ is a positive integer.
\end{theorem}

Using Theorem $9.1$, we construct locally recoverable codes in the following proposition.
\begin{proposition}
Let $q, m$, $F_{X}:=F$, $F_{Y}:=\mathbb{F}_{p^s}(x)$ and $\psi:=\phi$ as in section $3$. Let $S:=\{P_{\alpha_{1}},\cdots,P_{\alpha_{m}}\}$ and $A:=l P_{\infty}$, $l \in \mathbb{N}$ and $P_{\infty}$ is the infinite place of $F_{Y}$. If $2m-lq$ is a positive integer, then there exists a linear $(n,k, r_{0})$ LRC code $C(A, \psi)$ with parameters 
$$n=qm,~k \geq (q-1)(l+1) \text{ and } d \geq 2m-lq.$$  If $q=2$, then the code is optimal with locality $r_{0}=1$.
\end{proposition}
\begin{proof}
In terms of above notations, we have $h:=m,~s_{0}:=m,~r_{0}:=q-1,~n:=qm, ~g_{Y}=0$. So, by Theorem $9.1$, we get a linear $(n,k, r_{0})$ LRC code $C(A, \psi)$ with parameters $n=qm,~k \geq (q-1)(l+1) \text{ and } d \geq 2m-lq.$
Now
\begin{align*}
d+k+\frac{k}{r_{0}} &\geq 2m-lq+q(l+1)\\
&=n+2-(m-1)(q-2).
\end{align*}
So, when $q=2$, we get an optimal locally recoverable code with locality $r_{0}=1$.
\end{proof}

\begin{example}
Let $K=\mathbb{F}_{9}$. Let $q=3$, $m=4$, $l=2$, then we get $(12, \geq 6,2)$ locally recoverable code with minimum distance $d\geq 2$. Here, we have $d+k+\lceil \frac{k}{r_{0}} \rceil \geq 11$, while for code which meets equality in the equation $(9.1)$, we have $d+k+\lceil\frac{k}{r_{0}} \rceil=13$. This code is not optimal locally recoverable, but the gap is small.\par
Similarly, we get an LRC code with parameters $(15, \geq 6,2)$ with $d \geq 4$. In this case, also the gap in equation $(9.1)$ is small.
\end{example}

\begin{corollary}
Assume that all the hypotheses of Proposition $9.2$ holds. If $l \geq m-1$, then there exists a linear $(n,k, r_{0})$ LRC code $C(A, \psi)$ with parameters 
$$n=qm,~k \geq (q-1)(l+1) \text{ and } d \geq 2m-lq.$$ Also, equality holds in equation $(9.2)$.
\end{corollary}
\begin{proof}
The first part of corollary follows from Proposition $9.2$. Now, from equation $(9.2)$
\begin{align*}
\frac{(q-1)}{q} \geq \frac{k}{n} \geq \frac{(q-1)(l+1)}{qm} \geq \frac{(q-1)}{q}.
\end{align*}
So, we obtain codes with highest rate.
\end{proof}

\section*{\textbf{Concluding remarks}}
In this note, we have defined one-point geometric Goppa codes from an elementary abelian $p$-extension of $\mathbb{F}_{p^{s}}(x)$ and determined their dimension and the exact minimum distance in a few cases. Also, we have listed the exact second generalized Hamming weight of these codes in a few cases. We have also given simple criteria for quasi-self-duality of one-point geometric Goppa codes and self-duality of geometric Goppa codes with a divisor $G$ (not necessarily one-point). We have obtained families of quantum codes and convolutional codes from constructed one-point geometric Goppa codes. We have also obtained locally recoverable codes from $F/\mathbb{F}_{p^s}$. It will be interesting to calculate the higher generalized Hamming weights of these codes and other classes of weak Castle codes.

\section*{\textbf{Acknowledgement}}
The authors are very grateful to the anonymous reviewer for his/her comments and suggestions which help to improve the quality of the note. The second named author is supported by Early Career Research Award
(ECR/2016/000649) by the Department of Science \& Technology (DST), Government of India.

\bibliographystyle{plain}

\end{document}